\documentclass[runningheads]{lipics-v2021}

\hideLIPIcs  

\nolinenumbers
\usepackage{macros}

\title{Drawing Reeb Graphs}

\author{Erin Chambers}
    {University of Notre Dame, Notre Dame, IN, USA}
    {echambe2@nd.edu}%
    {https://orcid.org/0000-0001-8333-3676}%
    {NFS CCF 2444309}

\author{Brittany Terese Fasy}%
    {Montana State University, Bozeman, MT, USA \and \url{http://www.fasy.us} }%
    {brittany.fasy@montana.edu}%
    {https://orcid.org/0000-0003-1908-0154}%
    {NSF CCF 2046730}

\author{Erfan {Hosseini Sereshgi}}
    {Tulane University, New Orleans, LA, USA}
    {shosseinisereshgi@tulane.edu}%
    {https://orcid.org/0000-0003-2548-7428}%
    {NSF CCF 2107434}

\author{Maarten L\"offler}
    {Utrecht University, the Netherlands}
    {m.loffler@uu.nl}%
    {https://orcid.org/0009-0001-9403-8856}%
    {}

\authorrunning{Chambers, Fasy, Hosseini Sereshgi, and L\"offler} 

\Copyright{Erin Chambers, Brittany Terese Fasy, Erfan Hosseini Sereshgi, and
Maarten L\"offler} 



\category{} 

\relatedversion{A preliminary abstract of this work was presented as a poster at GD 2023 and a short version is accepted at IWOCA 2025.} 

\acknowledgements{We thank Sarah Percival for her participation in an earlier version of this work. This research began as a working group at the Annotated
Signatures workshop in March 2023, funded by NSF CCF 2107434.}

\EventEditors{John Q. Open and Joan R. Access}
\EventNoEds{2}
\EventLongTitle{42nd Conference on Very Important Topics (CVIT 2016)}
\EventShortTitle{CVIT 2016}
\EventAcronym{CVIT}
\EventYear{2016}
\EventDate{December 24--27, 2016}
\EventLocation{Little Whinging, United Kingdom}
\EventLogo{}
\SeriesVolume{42}
\ArticleNo{23}

\begin{document}

\maketitle              

\begin{abstract}
    Reeb graphs are simple topological descriptors with applications in many
areas like topological data analysis and computational geometry. Despite their
prevalence, visualization of Reeb graphs has received less attention. In this
paper, we bridge an essential gap in the literature by exploring the complexity
of drawing Reeb graphs. Specifically, we demonstrate that Reeb graph crossing
number minimization is NP-hard, both for straight-lined and curved edges. On the other hand, we identify specific classes of Reeb graphs, namely
paths and caterpillars, for which crossing-free drawings exist. We also give an
optimal algorithm for drawing cycle-shaped Reeb graphs with the least number of
crossings and provide initial observations on the complexities of drawing
multi-cycle Reeb graphs. We hope that this work establishes the foundation for
an understanding of the graph drawing challenges inherent in Reeb graph
visualization and paves the way for future work in this area.

\end{abstract}

\section{Introduction}\label{sec:intro}
Reeb graphs have become an important tool in computational topology for
visualizing continuous functions on complex spaces as a simplified
discrete structure.  Essentially, these graphs capture how level sets of a
function from a topological space to the real numbers
evolve and connect. In a Reeb graph, the connected components of each level
set become points, forming a graph, where the vertices are found at heights
where topological changes occur and edges represent the evolution of a single
connected component; see \figref{reeb-example}.
\begin{figure}
    \centering
    \includegraphics[height=2.5in]{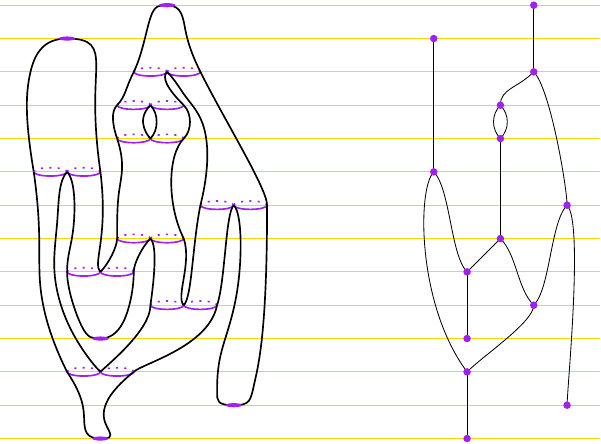}
    \caption{A two-manifold and its Reeb graph. }\label{fig:reeb-example}
\end{figure}
Originally introduced in~\cite{greeb}, recent work on computing Reeb graph
efficiently~\cite{Harvey2010,Parsa2012} has led to their increasing use in areas
such as graphics, visualization, and shape analysis~\cite{Biasotti2008,Yan2021}.

Despite how prevalent the use of Reeb graphs is, surprisingly little work has
been done on generating nice drawings of these structures.  The only work in
this area we are aware of considers book embeddings of these
graphs~\cite{Kurlin2013b}; while of combinatorial interest, the algorithms seem
less practical for easy viewing of larger Reeb graphs.  So, this leads to a very
natural question, both theoretical and practical: how difficult is it to draw
Reeb graph with minimal crossing?

Given the structural properties of Reeb graphs, there is an obvious connection
to level drawings of graphs, because Reeb functions are simply real-valued labels on
the vertices.
{\em Layered graph drawing}, also known as {\em hierarchical} or {\em
Sugiyama-style} graph drawing
is a type of graph drawing in which the vertices
of a directed graph are drawn in horizontal rows, or {\em layers}, an in which
edges are $y$-monotone curves~\cite{Bastert2001,bett98,sugiyama81}; depending on
the application, vertices may or may not be pre-assigned to levels.  As with
many drawing styles, layered drawings tend to be more readable when the number
of crossings is low; hence a significant amount of effort has been directed
towards crossing minimization in layered drawings of graphs. When a drawing
without any crossings is possible, the graph is called {\em upward planar} (when
vertices are not pre-assigned to levels) or
{\em level planar} (when they are).
Other work characterizes forbidden structures in level planar graphs, concluding
that there are infinitely many such forbidden minors in this class of level
planar graphs~\cite {EstrellaBalderrama2009OnTC}.
While it is possible to test whether a given directed
acyclic graph with a single source and a single sink admits an upward planar
drawing~\cite {garg95}, the problem becomes NP-hard when there can be multiple
sources or sinks~\cite{doi:10.1137/S0097539794277123}.  When the layers are
preassigned, the problem is easier, and testing whether a graph admits a {\em
level planar} embedding is possible in linear time~\cite
{10.1007/3-540-46648-7_7}. Recent work points out that several recent algorithms
for this are incorrect~\cite
{fprs-lpmdwt-24}; even in this
setting, the details can be very tricky and the exact status in unclear.
 However, when a planar embedding is not possible, the
problem of minimizing crossings is still of interest, and this problem is
NP-hard, even when there are only two levels~\cite {garey-johnson}.

Building on our previous poster paper~\cite{chambers2023drawing}, this work
initiates the study of crossing minimization in drawings of Reeb graphs, working
towards practical solutions to
the key problem of drawing these essential topological structures
effectively. In particular, we demonstrate the inherent complexity of this task
by proving that minimizing edge crossings is NP-hard. However, we make promising steps towards
parameterization by cycle count, which is expected to be low in many applications.
Specifically, we explore
the structural properties and drawing algorithms for acyclic and single-cycle
Reeb graphs.

\section{Preliminaries}\label{sec:background}
We lay out the foundational definitions of general and generic Reeb graphs in
\secref{reeb-graphs}, allowing us to define
the problem of minimum crossing number for Reeb graphs formally in
\secref{drawing-reeb-graphs}. Finally, we argue that we can focus our inquiry to
the more specialized domain of refined Reeb graphs in
\secref{refined-generic-reeb-graphs}.

\subsection {Reeb Graphs} \label{sec:reeb-graphs}

A Reeb graph, as first explained in~\cite{greeb}, is a mathematical construct
that represents the topological structure of a scalar function defined on a
manifold.
At a high level, the Reeb graph is a graph that encodes the connected components
of level sets and the critical points of a function defined over a manifold.

To define the Reeb graph, let $\mfd$ be a compact,
connected, and o-minimal
manifold of dimension $n$, and let $h: \mfd
\rightarrow \R$ be a smooth function. Define an equivalence relation on the
points of $\mfd$
by setting~$y \sim y'$ if $y$ and $y'$ are in the same
path component of $h^{-1}(a)$ for some~$a \in \R$.
Assuming the starting space and function are well-behaved (i.e.,~Morse functions
on manifolds, or tame functions on constructible spaces
\cite{Edelsbrunner2008a,deSilva2016}), the quotient space~$G_h = \mfd /\sim$ is
a finite graph~$G$
with the following structure:
The vertices of $G$ correspond to critical points of~$h$, that is,
the points~$p \in \mfd$ whose differential is zero ($dh_p=0$).
The critical points
of $h$ partition $\mfd$ into regions called cells, where each cell is a locally
path-connected compact space.
Two vertices are connected by an edge in~$G$ if their corresponding critical
points lie at the boundary of a common cell.
A \emph{Reeb graph} is the pair~$(G,h)$.

In graph theoretical terms, this means we have a graph $G = (V,E)$, where the
vertex set~$V$ comes with a {\em height} function $h : V \to \R$.

Let the degree $\deg(v)$ of a vertex be the number of edges in $E$ with $v$ as an endpoint.
Furthermore, we define the {\em downward degree} $\deg_{\downarrow} (v)$ as the
number of edges that connect~$v$ to a vertex of lower
height:~$\deg_{\downarrow} (v) = \{ w \in V | (v,w) \in E, h(w) < h(v)\}$,
and similarly the {\em upward degree} $\deg_{\uparrow} (v)$ as the number of edges
that connect $v$ to a vertex of higher
height:~$\deg_{\uparrow} (v) = \{ w \in V | (v,w) \in E, h(w) > h(v)\}$.
Note that $\deg_{\downarrow} (v) + \deg_{\uparrow} (v)
= \deg(v)$ for all~$v \in V$.

We define a \emph{generic Reeb graph} as a Reeb graph $(G,h)$ such that:  every
vertex has a unique height, so $h(v) \neq h(w)$ for any $v \neq w$;  every
vertex has $\deg(v) = 1$ or $\deg(v) = 3$; and furthermore $\deg_\downarrow (v)
\leq 2$ and $\deg_\uparrow (v) \leq 2$.

\subsection {Drawings of Reeb Graphs} \label{sec:drawing-reeb-graphs}

A {\em drawing} of a Reeb graph is an immersion of $G$ into $\R^2$. We have the
following constraints on this immersion:
\begin {itemize}
    \item Each vertex $v \in V$ gets mapped to a unique point $(x,y) \in \R^2$,
        with the additional requirement that $y = h(v)$.
    \item Each edge $e=(v,w) \in E$ gets mapped to a continuous $y$-monotone
        curve with $v$ and $w$ as endpoints.
\end {itemize}

This problem definition is somewhat related to the notions of {\em upwards}
drawings and {\em level} drawings, given the height functions and edge
constraint, although the concepts are not quite identical.  In particular, in
upwards planar, the exact height is not fixed for a vertex, but rather all
(directed) edges must orient upwards.  Level planar drawings generally fix
integer-valued heights, and not real-valued, and algorithms in this setting
generally seek drawings with no crossing, rather than attempting to minimize
crossings as we do in this paper.

Minimizing edge crossings in a graph drawing is a desirable aesthetic criterion.
This leads to the computational problem below:
\begin{problem}[Reeb Graph Crossing Number]
    Given a Reeb graph, the Reeb graph crossing number (RGCN) is the
    minimum number of crossings among all drawings of that Reeb graph.
\end{problem}
Ultimately, finding a drawing realizing the minimum number of crossings is the
goal.
Computing the crossing number is NP-hard in general~\cite{garey-johnson}, but to
the best of our knowledge the Reeb crossing number has not been studied
previously.  Notably, the proof by Garey and Johnson~\cite{garey-johnson} does
not extend to
this special case, as it uses high-degree vertices.  Note that our proof does
not rely on vertices having unique heights, but does ensure low degree vertices.

\subsection {\Subdivided 
Reeb Graphs} \label{sec:refined-generic-reeb-graphs}

Given a generic Reeb graph $(G,h)$, consider the finite set of heights. These
heights define a set of `levels' in the Reeb graph.
A Reeb graph or generic Reeb graph may have edges between vertices at non-consecutive levels.
By subdividing any longer edges so that all edges are between consecutive
levels, we arrive at a \subdivided Reeb graph:

\begin {definition}[\Subdivided Reeb Graph]

    A {\em \subdivided 
    Reeb graph} $(G=(V,E),h)$ is a Reeb graph with the~properties:
    \begin {itemize}
        \item every edge $e = (v,w)$ in $E$ is between consecutive levels
        \item $deg_\downarrow (v) \leq 2$ and $\deg_\uparrow (v) \leq 2$ and $\deg(v)
            \leq 3$ (note: but no longer necessarily $\deg(v) = 1$ or $3$, it can
            also be $2$).
    \end {itemize}

\end {definition}

We next prove that a generic Reeb graph
and its associated \subdivided Reeb graph have the same crossing complexity.

\begin{lemma}[Constructing a \Subdivided Reeb Graph]
    For every generic Reeb graph, there is a \subdivided 
    Reeb graph, and
    their optimal drawings have the same number of crossings.
\end{lemma}

\begin {proof}
    We first provide an explicit construction.
    Let $G$ be a generic Reeb graph. We construct a \subdivided 
    Reeb graph $G'$ as follows.
    Let $v_1, v_2, \ldots, v_n$ be the vertices of $G$ sorted by $h$; as the
    graph is generic, we know heights are unique.
    For every vertex $v_i$ of $G$, we create a vertex $v_i'$ of~$G'$ and
    set~$h(v_i')=i$.
    Then, for every edge $v_iv_j$ of~$G$ with $i < j$, we subdivide it by adding
    a vertex $x_k$ for each~$k \in [i+1, j-1]$. Hence, the edge~$v_iv_j$ becomes
    the path~$v'_i
    x_{i+1} x_{i+2} \ldots x_{j-1} v'_j$.

    We now argue that optimal drawings of $G$ and $G'$ have the same number of
    crossings. Clearly, a drawing $\Gamma$ of $G$ can be transformed into a
    drawing $\Gamma'$ of~$G'$ by simply drawing a horizontal line through each
    vertex of $\Gamma$ and scaling the strips between consecutive lines such
    that their height becomes $1$; because edges are $y$-monotone curves, they
    only cross each strip once. We draw additional vertices where necessary on
    the strip boundaries, subdividing longer edges at these heights. This does
    not change the number of crossings, which occurs inside of the strips.
    We note that this transformation is reversible, as we can ``un-subdivide''
    edges and rescale the vertices, so a drawing of $\Gamma'$ also corresponds
    to a drawing of $\Gamma$ under this reverse transformation.
\end {proof}

\section{Straight-Edge or Curved Drawings}

\label{sec:straight-or-curved}

A well-known result in graph theory is that a given graph has a plane embedding if
and only if it has a straight-line plane embedding.
We show here that the same is true for Reeb graphs without crossings.
However, the same is not true for drawings with minimal crossing Reeb graphs:
requiring the edges to be straight line segments may result in a
higher crossing~number.

\begin{lemma}
    Suppose we are given a planar (zero-crossings) drawing of a \subdivided,
    generic Reeb graph~$G$, where every edge
    is an $xy$-monotone curve and every vertex is at its specified~$y$-coordinate.
    Then, there exists another planar drawing of $G$, where the vertices are in
    the same order in the $x$-dimension, the vertices are at the same
    $y$-coordinates, and the edges are drawn as straight line segments.
    \label{lem:straight}
\end{lemma}

To prove this, we adapt the ``shift method'' of~\cite{DeFraysseix1990}, which is
for planar triangulations but can be adapted to our setting.
See~\figref{obs-curves-are-straight} for an example of this process.

\begin{proof}

    Begin with a drawing of $G = (V,E)$ with $xy$-monotone curves for the edges
    in $E$.
    We know that the edges of the drawing are in a partial order $P(E)$ based on
    their left-right relationship.  Therefore, for any two edges that share a
    vertical span, because the drawing is non-crossing, one edge must be to the
    left of the other, and this relation is acyclic, because the drawing is
    planar.

    Next we construct a total order $T(V)$ on the vertices from left to right.
    We construct this order by inserting vertices as follows: first, add the
    leftmost vertex (leftmost lowest in case of ties). Then, we say a vertex is
    {\em free} if all its incident edges to previous vertices in the order have
    the property that all edges that come before it in the partial order $P(E)$
    are between vertices that were already added before.
    There is always at least one free vertex, and we may choose any free vertex
    to insert into $T(V)$. Repeat until we have inserted all vertices.

    Now, process the vertices in the order $T(V)$. Whenever we process a new
    vertex, by construction, all vertices it needs to connect to are currently
    exposed to the right. Therefore, we can choose its $x$-coordinate
    sufficiently far to the right so that its edges do not intersect~anything.
\end{proof}

    \begin{figure}[htb]
        \centering
        \includegraphics[height=1.2in]{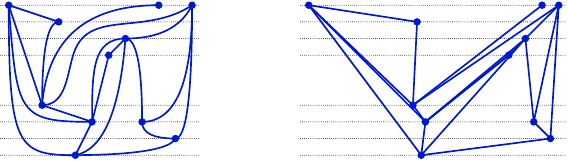}
        \caption{Two level drawings of the same graph, one with $y$-monotone curves
        and one with straight line segments.
        }
        \label{fig:obs-curves-are-straight}
    \end{figure}

However, when there are crossings present in the curved drawing, a similar
statement does not hold.

\begin{lemma}[Non-Stretchable Graph Drawing]\label{lem:nonstretchable}
    There exists a Reeb graph $(G=(V,E),h)$ and a drawing $\Gamma$ of $G$ such that:
        every vertex $v \in V$ has $y$-coordinate $h(v)$ in $\Gamma$;
       every edge is an $xy$-monotone curve in $\Gamma$; $\Gamma$ has $k$ crossings;
       and there exists no straight-line drawing of~$G$ with $\leq k$ crossings.

\end{lemma}

\begin{proof}
    We prove this by providing a graph drawing that cannot be drawn with
    straight lines using less crossings. We start with the well-known
    nonstretchable pseudoline arrangement in
    \subfigref{obs-curves-are-not-straight}{lines},
    see~\cite[Fig.~5.3.2]{hcg-pseudolines} and the associated proof that this
    line arrangement cannot be drawn with straight lines.
    Instead of thinking of these as infinite pseudolines, we view this as an
    arrangement of nine $xy$-monotone edges and connect the endpoints of the edges
    `in order' as shown with the black edges of
    \subfigref{obs-curves-are-not-straight}{graph}.
    Next, we add the dual graph to this arrangement, shown in red. For each
    `outermost' red vertex, we add edges to the pseudoline endpoints on the
    boundary of the two-cell containing that vertex (shown in green). Finally,
    we wrap this construction with 20 copies of the black graph, each slightly
    larger than the last (shown in orange), and connecting consecutive copies of
    a vertex by a single edge.

    \begin{figure}[h!t!b!]
        \centering
        \begin{subfigure}[b]{0.45\textwidth}
            \includegraphics[height=2in]{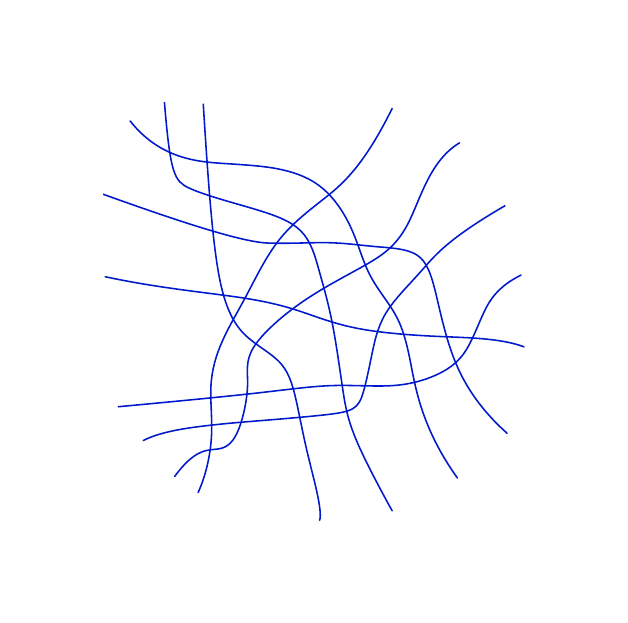}
            \caption{Non-Stretchable Arrangement}\label{fig:obs-curves-are-not-straight-lines}
        \end{subfigure}
        \begin{subfigure}[b]{0.45\textwidth}
            \includegraphics[height=2in]{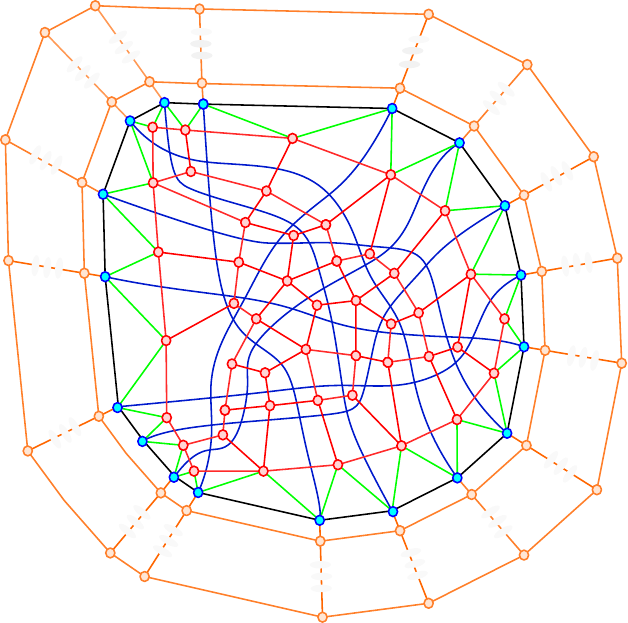}
            \caption{Non-Stretchable Graph Drawing}\label{fig:obs-curves-are-not-straight-graph}
        \end{subfigure}
        \caption{
            Building a non-stretchable graph drawing \lemref{nonstretchable}).
            \subonlyref{obs-curves-are-not-straight-lines}
            This arrangement of $xy$-monotone pseudoline is not stretchable, which
            means it cannot be drawn with straight lines.
            \subonlyref{obs-curves-are-not-straight-graph}
            This graph contains the pseudoline arrangement and its dual as
            sub-structures.
        }\label{fig:obs-curves-are-not-straight}
   \end{figure}
   We count the crossings: there are $36$ blue crossings,
   and each blue curve crosses exactly nine red edges.  As the red graph is dual to the blue graph, any other way
   to draw these two graphs results in more than nine crossings per blue
   edge, so we cannot change the combinatorial
   embedding of the blue arrangement without increasing the number of~crossings.
   \journal{this argument is a bit hand-wavy ...}
\end{proof}

\section{NP-Hardness}\label{sec:np-hard}
In this section, we prove that the drawing problem for \subdivided generic Reeb
graphs is NP-hard.
\begin{theorem}[RGCN is NP-Hard]\label{thm:NP-hardness}
    Given a generic Reeb graph $(G=(V,E),h)$,
    and an integer $k$, deciding if the
    crossing number is less than $k$ is~NP-hard.
\end{theorem}

Our proof is inspired by the proof of Garey and Johnson~\cite{garey-johnson};
however, we require several new ideas. Before we explore the NP-hardness proof,
we explain the problem we use to prove NP-hardness and describe a crucial
structure in the proof in detail.  We transform a known NP-complete problem,
optimal linear
arrangement (OLA)~\cite{GAREY1976}, to the generic Reeb graph crossing number
problem.
We interpret this as ``arranging'' the vertices of the graph in a linear
order, such that the total length of all the edges of the graph is
less than $k$.
Formally, the OLA problem~\cite{GAREY1976} is defined as follows:
Given a graph $G=(V,E)$ and an integer $k$, is there a function $f: V
\rightarrow \{1,2,...,|V|\}$ such that $\sum_{\{v,w\}\in E} |f(v)-f(w)| \leq k$.

\subsection{Triangular Hexagonal Grid}

We construct a triangular hexagonal grid $T_k$ using vertices and edges where
$k$ is the number of rows in the grid and each row $i$ consists of $i$ hexagons.
A grid~$T_k$ consists of $k^2 + 4k + 1$ vertices and $\frac{3}{2}(k^2+3k)$
edges. The use of hexagonal grids maintains the characteristics of generic Reeb
graphs, such as maximum vertex
degree, while safeguarding the graph structure against changes caused by
swapping vertices along the vertical axis. This precaution is essential as such
alterations would result in a greater number of crossings per swaps; see \figref{thg}.
\begin{figure}[hbtp]
    \centering
    \includegraphics[height=1.5in]{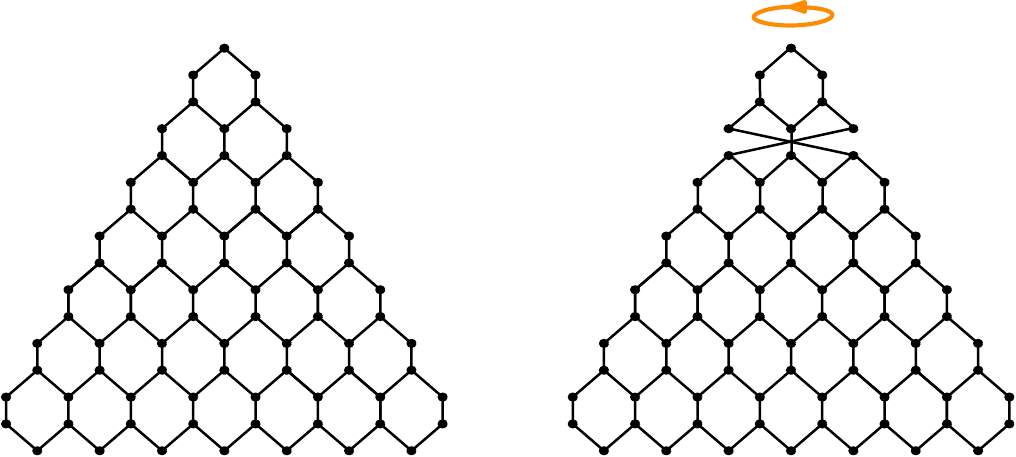}
    \caption{A triangular hexagonal grid on the left and the same grid if we
        twist the structure vertically once over the second row.}
    \label{fig:thg}
\end{figure}
Therefore, we use it in \thmref{NP-hardness} to build our gadget for the
NP-hardness~proof.
\journal{We should argue somewhere why each hexagonal grid must be embedded
without crossings. All of the arguments above assume that these are somehow
``rigid'' blocks...}

\subsection{Construction} \label{sec:construction}
\begin{figure}[hbtp]
    \centering
    \includegraphics[height=1.2in]{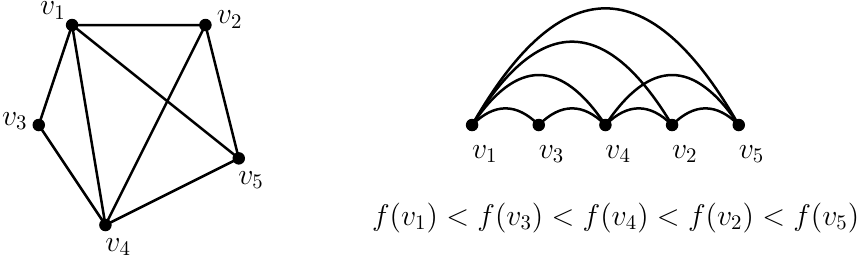}
    \caption{Graph on the left and its image based on an optimal order
    defined by a function~$f: V \rightarrow \{1,2,...,|V|\}$ on the right.
    }
    \label{fig:ex-graph-nphardproof}
\end{figure}
Suppose we have an instance of the OLA problem, given
a graph $G = (V,E)$ and $k$; see
\figref{ex-graph-nphardproof} for an example.
We construct a generic Reeb graph $H= (V_1 \cup V_2, E_1 \cup E_2 \cup E_3
\cup E_4)$ such that:
We replicate every vertex in the set $V$ to form the set $V'$.
The vertices
of~$V$ are positioned horizontally in an arbitrary order as $v_1,v_2,
\ldots v_n$
and the same ordering is applied to the vertices of $V'$. Subsequently, all
vertices in both $V$ and $V'$ are
transformed into triangular hexagonal grids $T_{|E|}$, each consisting of~$|E|^2+4|E|+1$ vertices
with $|E|^2+1$ vertices at the bottom of the triangular structure.  We then
employ the middle vertex as a connector (analogous to an extension cord) to
facilitate the drawing of edges of $H$.  Furthermore, we establish edges
between the corresponding $|E|^2$ vertices for each pair of hexagonal grids.

\figref{ex-graph-structure-nphardproof} shows a construction based on an
arbitrary order of the vertices, and
\figref{ex-graph-structure-nphardproof-sol} demonstrates a solution for our
problem given a solution for OLA.

\begin{itemize}
    \item $V_1 = \{|E|^2+4|E|+1 \text{ vertices}\}$ 
    \item $V_2 = \{2 \cdot \deg(v_i) \cdot |V| \text{ copies of } v_i \in V $\} 
    \item $E_1 = \{(v_i,w'_j) | v_i,w'_j \in V_2, \ (v_i,w_j) \in E, \ i>j \}$ 
    \item $E_2 = \{|E|^2 \text{ copies of }  \{v_i,v'_i\}; v_i \in V, v'_i \in V'\}$ 
    \item $E_3 = \{(v_{i_p},v_{i_{p+1}})| v_{i_p},v_{i_{p+1}} \in V_2 \}$ 
    \item $E_4 = \{\frac{3}{2}(|E|^2+|E|) \text{ edges}\}$ 
    \item $k' = |E|^2(k-|E|)+(|E|^2-1)$
\end{itemize}
Note that constructing $H$ and $k'$ from $G$ and $k$ can be done in
polynomial time. Additionally,~$H$ is a
connected graph only if $G$ is connected.

\subsection {Correctness}

Next, we argue that $H$ has an embedding with at most $k'$ crossings if and
only if $G$ has a linear arrangement of cost at most $k$.

\begin{figure}[hbtp]
    \centering
    \includegraphics[width=0.8\linewidth]{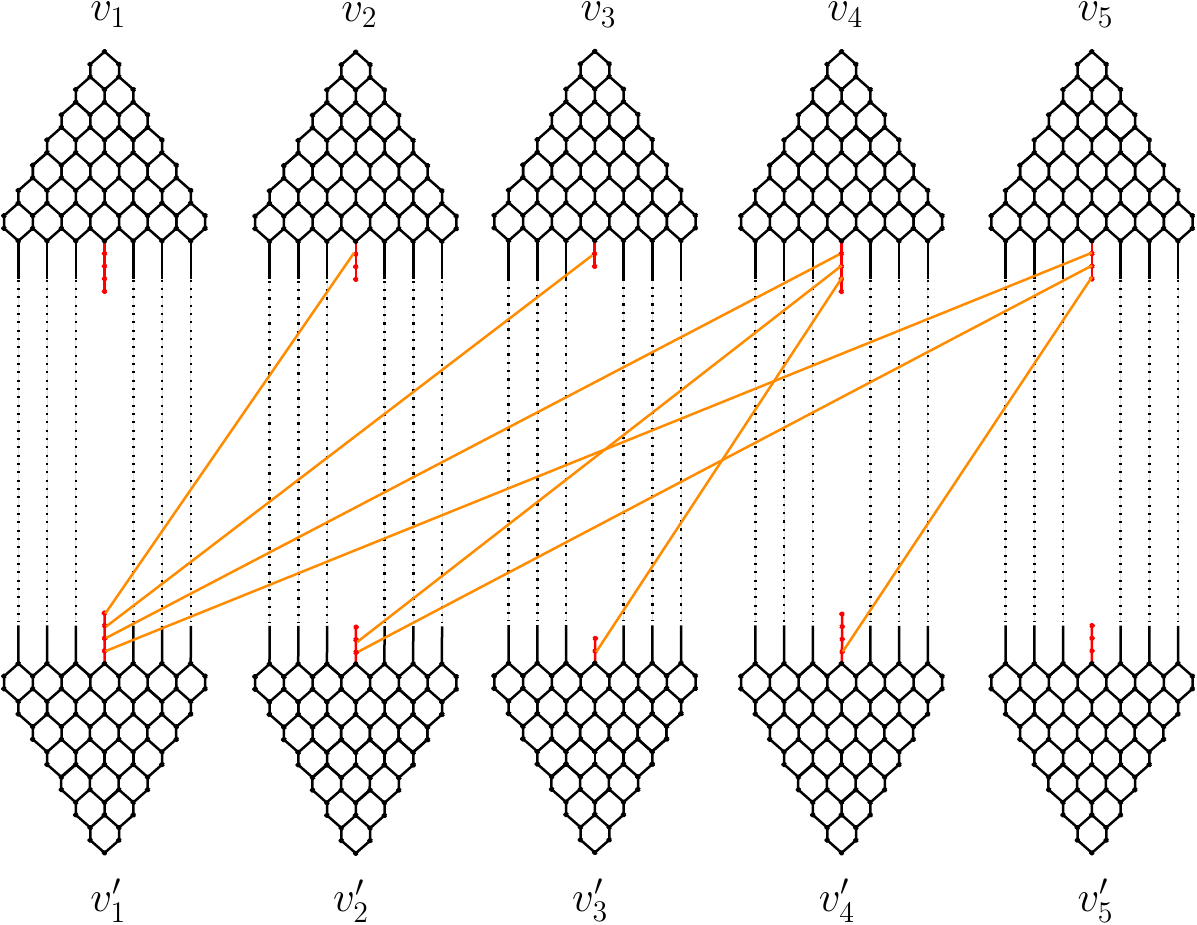}
    \caption{The corresponding Reeb graph construction of
    \figref{ex-graph-nphardproof} in our proof. Black vertices are~$V_1$ and
    red vertices are $V_2$. $E_1$ edges are shown in orange, $E_2$ edges are
    dotted, $E_3$ are in red and~$E_4$ are the edges of the triangular hexagonal
    grids represented in black. }
    \label{fig:ex-graph-structure-nphardproof}
\end{figure}

\begin{figure}[hbtp]
    \centering
    \includegraphics[width=0.8\linewidth]{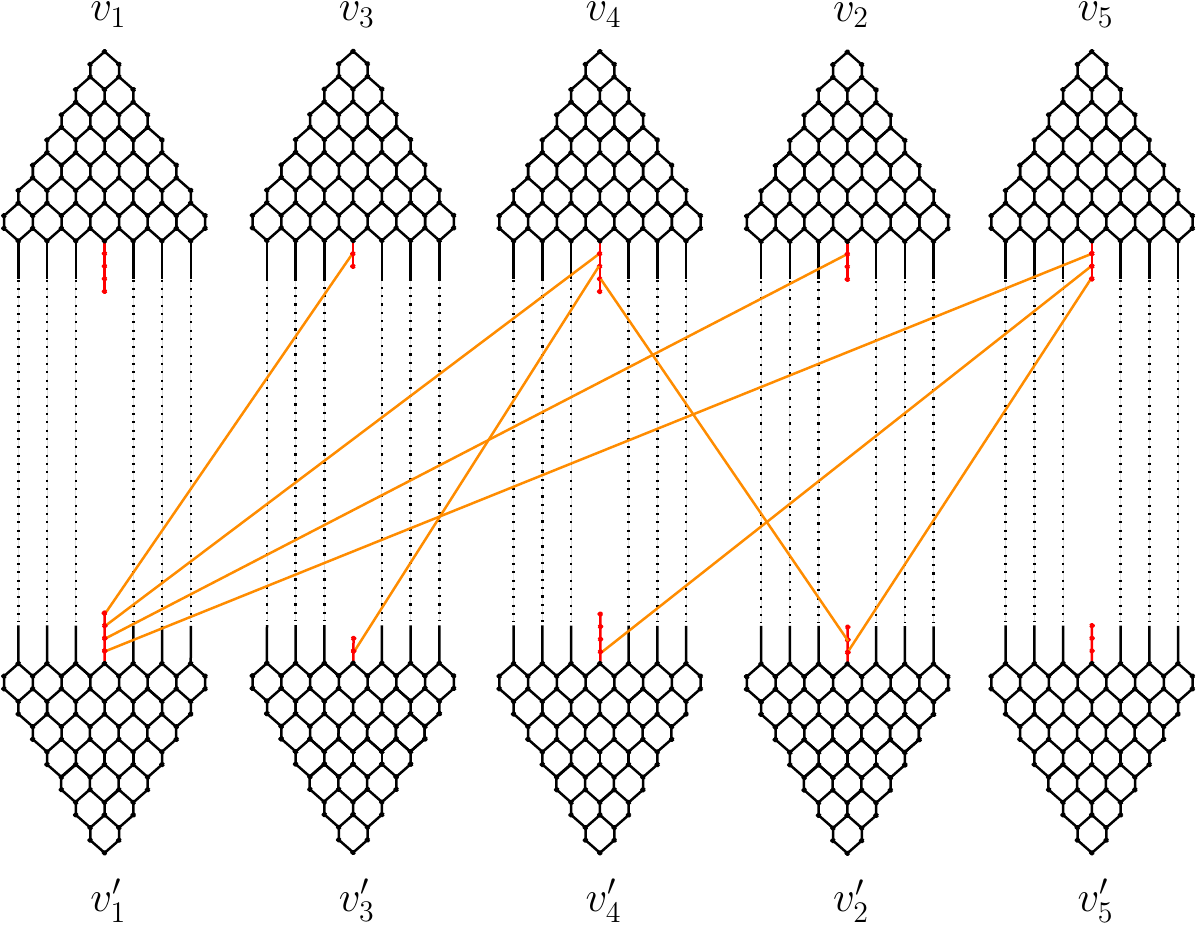}
    \caption{The layout of the Reeb graph of
        \figref{ex-graph-structure-nphardproof} given a solution to OLA problem.
        Black vertices are $V_1$ and red vertices are $V_2$. $E_1$ edges are shown
        in orange, $E_2$ edges are dotted, $E_3$ are in red and $E_4$ are the edges
        of the triangular hexagonal grids represented in black.}
    \label{fig:ex-graph-structure-nphardproof-sol}
\end{figure}

\begin{lemma} \label{lem:G-implies-H}
    A valid linear arrangement of $G$ implies a valid embedding of
    $H$.
\end{lemma}
\begin{proof}
    Assuming that there exists a function $f: V \rightarrow \{1,2,...,|V|\}$
    satisfying $\sum_{\{v,w\}\in E} |f(v)-f(w)| \leq k$, we utilize this order
    to construct $H$ as illustrated in
    \figref{ex-graph-structure-nphardproof}. Each edge $\{v_i,w_j\}$ in $E_1$
    intersects $(|f(v_i)-f(w_j)| -1) \cdot |E|^2$ edges in $E_2$. Hence, the
    total number of crossings between edges in $E_1$ and edges in $E_2$ is at
    most $\sum_{\{v,w\} \in E}(|f(v)-f(w)| -1) \cdot |E|^2 \leq (k-|E|) \cdot
    |E|^2$. Considering that the total number of crossings among edges in $E_1$
    is at most $(|E|^2-1)$, the overall number of crossings is at most
    $|E|^2(k-|E|)+(|E|^2-1)$, which is equal to $k'$.
\end{proof}

\begin{lemma}  \label{lem:H-implies-G}
    A valid embedding of $H$ implies a valid linear arrangement of $G$.
\end{lemma}
\begin{proof}
    Now, assume that an embedding of $H$ is provided. From this embedding,
    we can extract two distinct one-to-one functions, $f_1$ and $f_2$, mapping
    the vertices $V$ to the set~$\{1,2, \ldots ,|V|\}$ based on the order of the top
    and bottom rows of the corresponding hexagonal grids, respectively. Our goal
    is to prove that these two functions are identical. Consider two vertices
    $v$ and $w$ in $V$ such that $f_1(v) < f_1(w)$ and $f_2(v') > f_2(w')$. In
    this case, there are $|E|^2$ edges between $v$ and $v'$ in $E_2$ that
    intersect $|E|^2$ edges between $w$ and $w'$ in~$E_2$, resulting in a total
    of $|E|^4$ crossings. This contradicts the assumed bound on the number of
    crossings, which means that $f_1$ and $f_2$ must indeed be identical.
    Consequently, all vertices~$v$ in $V$ face their respective corresponding
    vertices $v'$ as depicted in
    \figref{ex-graph-structure-nphardproof}.
    Therefore, each edge $\{v_i,w'_j\}$ in $E_1$ crosses $(|f_1(v_i)-f_1(w_j)|
    -1) \cdot |E|^2$ edges of $E_2$. From these observations, we have:
    \begin{equation}
        \begin{split}
            \sum_{\{v,w\} \in E} &(|f_1(v)-f_1(w)| -1).|E|^2 \leq k' \\
                & \Rightarrow
                \sum_{\{v,w\} \in E}(|f_1(v)-f_1(w)| -1).|E|^2 \leq
                |E|^2(k-|E|)+(|E|^2-1)\\
                & \Rightarrow
                \sum_{\{v,w\} \in E}(|f_1(v)-f_1(w)| -1) \leq (k-|E|)\\
                & \Rightarrow
                \sum_{\{v,w\} \in E}|f_1(v)-f_1(w)| \leq k.
        \end{split}
    \end{equation}
\end{proof}

Now, we are ready to prove \thmref{NP-hardness} using the construction and the
as stated~above.
\begin{proof}[Proof of \thmref{NP-hardness}]
    Given an instance $G$ of OLA, we construct a graph $H$ as in
    \secref{construction}, by \lemref{G-implies-H}, a valid linear arrangement
    of $G$ implies a valid embedding of
    $H$, and by \lemref{H-implies-G}, the reverse is also true. Finally, note
    that $H$ has polynomial size \journal{fill in exact size?} and the
    construction of $H$ takes polynomial time.
\end{proof}

\section{Parameterized Complexity in the Number of Cycles}\label{sec:cycles}
The number of cycles in a Reeb graph for oriented surfaces is connected to its
genus. Therefore, it is of interest to investigate the
complexity of the drawing problem for generic Reeb graphs, parameterized by the
number of cycles in the graph.
In this section, we make first steps towards such an investigation by
considering several special graph classes with zero or one cycle.
Throughout this section, let $(G,h)$ be a Reeb graph; we seek a drawing that
minimizes the number of crossings.

\subsection{No Cycles}

Consider the case where $G$ has no cycles (i.e., $G$ is a tree).
If we did not have the constraint that each vertex has a prescribed height, then
$G$ is a planar graph and can hence be drawn without crossings.
If we instead consider the problem where edges are pointed in a particular
direction (upward or downward),
it is still possible to draw $G$ without crossings~\cite{DiBattista1998}.
However, once we require the drawing to respect~$h$, then it is possible that
drawing~$G$ requires crossings.  For example, see \figref{ex-tree-nonplanar}.

We start with a basic observation, which has been used in the literature
(e.g.,~\cite {EstrellaBalderrama2009OnTC}).
If~$G$ has no branching points (i.e., is a single path), then drawing~$(G,h)$
without crossings is always possible:

\begin{construction}[A Path]\label{construct:path}
    Let $(G,h)$ be a Reeb graph such that $G$ is a single path.  Then, label the
    vertices consecutively by the order they appear in the path:~$v_1,v_2,
    \ldots, v_n$.  We draw~$v_i$ at $(i,h(v_i))$, with straight-line edges
    connecting adjacent vertices on the path.
\end{construction}

    \begin{figure}[hbt]
        \begin{subfigure}[b]{0.45\textwidth}
            \centering
            \includegraphics[height=1in]{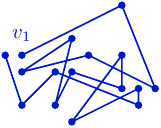}
            \caption{Before}
            \label{fig:ex-path-random}
        \end{subfigure}
        ~~
        \begin{subfigure}[b]{0.45\textwidth}
            \centering
            \includegraphics[height=1in]{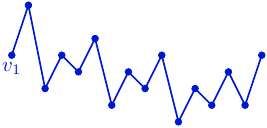}
            \caption{After}
            \label{fig:ex-path-lefttoright}
        \end{subfigure}
        \caption{Drawing the Reeb graph of a path.
            In \subonlyref{ex-path-random}, we see a drawing of the Reeb graph with
            seven crossings. We label one endpoint $v_1$, then draw
            consecutive vertices to the right to obtain the drawing without
            crossings shown in~\subonlyref{ex-path-lefttoright}.}
        \label{fig:ex-path}
    \end{figure}

See \figref{ex-path} for an illustration of Construction~\ref{construct:path}.
Because no two edges share an $x$-coordinate, other than their shared endpoint, we
see that this drawing has no crossings.
In fact, the same result also applies to {\em caterpillars}, see
\figref{ex-caterpillar-lefttoright}.

\begin{construction}[A Caterpillar]\label{construct:caterpillar}
    Let $(G,h)$ be a Reeb graph such that $G$ is a caterpillar. We apply
    Construction~\ref {construct:path} to the spine of $G$. Then, because each
    vertex $v$ of $G$ has $\deg(v) \le 3$, we can simply draw the legs straight up or down.
\end{construction}
\begin{figure}
    \centering
    \includegraphics[height=1in]{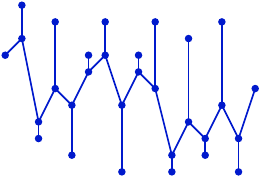}
    \caption{Every caterpillar can be embedded without crossings.}
    \label{fig:ex-caterpillar-lefttoright}
\end{figure}

Interestingly, because generic Reeb graphs have the property that all vertices
are degree one or three, this means all internal vertices of the spine have
exactly one leg; furthermore, when the spine is bending up the leg points down
and vice versa.
For more complicated trees (e.g., lobsters) it is no longer true that they can
always be drawn without crossings, and the problem of minimizing their crossings
remains open. Additionally, \figref{ex-tree-nonplanar} shows an example of a
planar graph---in particular, a tree---that cannot be drawn without
crossings while respecting the height function.

\begin{figure}[htb]
   \centering
    \includegraphics{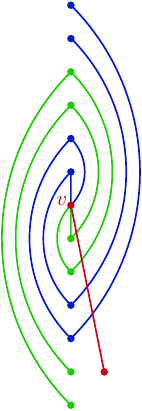}
    \caption{The red vertex $v$ has $\deg(v)=5$. We consider three subtrees rooted at $v$: the green one, the
    blue one, and the red one.
    In order to draw the blue and green subtrees without overlapping,
    they must `spiral' around the red vertex. However, by doing so, it is not
   possible to draw the red edge to be monotonically decreasing without
    crossing the blue subtree.}\label{fig:ex-tree-nonplanar}
\end{figure}

\subsection {One Cycle}

We next consider drawing $(G,h)$ when $G$ is a cycle (and nothing more).
Intuitively, we approach this type of graph by specifying the number of times
such a cycle must alternate up and down edges, as these are the potential
crossings.  We make this more precise as~follows; see
\figref{alg-onecycle-topdowniterationnumber} for an example.

\begin {definition} [Top-Down Iteration Number]\label{def:topdown}
      Let~$(G,h)$ be a Reeb graph and let $C$ be a cycle in $G$.
      Let $t = \max_{v \in C} h(v)$ be the {\em top} level used by $C$, and
      let~$b = \min_{v \in C} h(v)$ be the bottom level.
      Now, number the vertices of $C$ in order of the cycle, starting with some vertex $v_1$ such that $h(v_1) = t$.
      Set $t_1 = 1$, and then set~$b_i = \min_{j > t_i} \{ h(j)=b \}$
      and~$t_i = \min_{j > t_{i-1}} \{ h(j)=t \}$
      be the first alternating indices where the cycle uses the top and bottom level respectively.
      Let~$k$ be the total number of iterations.
      We say that $k$ is the {\em top-down-iteration-number} of~$C$.
\end{definition}

\begin{figure}[tbh]
    \centering
    \includegraphics[height=1in]{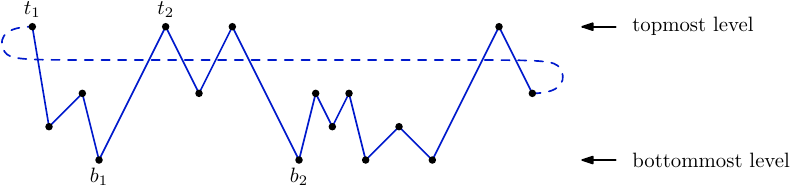}
    \caption{An example of a cycle with top-down-iteration-number equal to $2$.}
    \label{fig:alg-onecycle-topdowniterationnumber}
\end{figure}

We next state the main result we set out to prove in this section:

\begin{theorem} [Cycle Drawing] \label{thm:cycle-drawing}
    Let $G$ be a cycle, and let $k$ be the top-down iteration number of $G$.
      Then, there exists a drawing of $G$ with $k-1$ crossings.
\end{theorem}

Before proving the theorem, we first build some intuition.
When drawing a cycle, the main factor that influences the number of crossings is
the number of vertices at the highest or lowest levels.
We first consider two special cases: the case where there is a single
global maximum and a single global minimum, as in
\figref{onecycle-uniquetopbot}, and the case where {\em every} vertex is
a global maximum or minimum, as in \figref{ex-onecycle-bowtie}.

\begin{figure}[tbh]
\begin{subfigure}[b]{0.5\textwidth}
    \centering
    \includegraphics[height=1in]{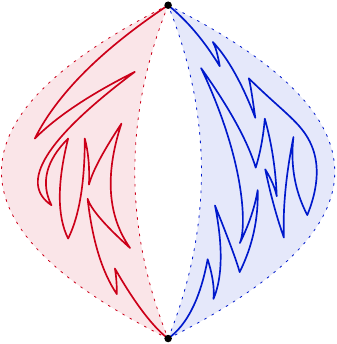}
    \caption{}
    \label{fig:onecycle-uniquetopbot}
\end{subfigure}
\begin{subfigure}[b]{0.5\textwidth}
    \centering
        \includegraphics[height=.7in]{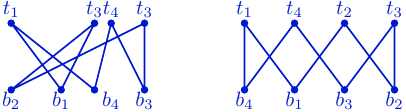}
        \caption{}\label{fig:ex-onecycle-bowtie}
\end{subfigure}
    \caption{
        (a) When there is a unique topmost and bottommost vertex, a cycle can
        always be drawn without crossings. (b) Drawing a cycle that alternates
        between two levels. (Left) A graph with the vertices labeled as
        described in \constructref{bowtie}. (Right) The bowtie layout.
    }
    \label{fig:enter-label}
\end{figure}

\begin {lemma}
    When $G$ is a cycle and there is a unique topmost and a unique bottommost
    vertex, then $G$ can be drawn without crossings.
\end{lemma}
\begin {proof}
    We apply Construction~\ref{construct:path} twice, laying each path out disjointly to one
    ``side'' of the drawing.
\end {proof}

We next consider what we call a bowtie graph, where every vertex is a global
maximum or minimum; see \figref{ex-onecycle-bowtie} for an illustration.

\begin{construction}[Bowtie]\label{construct:bowtie}
    If $G$ is a cycle with $2n$ vertices alternating between two heights, then the
    \emph{bowtie construction} is created as follows:
    \begin{enumerate}
        \item Choose an arbitrary vertex on the top level to start; call it
            $t_1=v_1$.
        \item Label the next vertex (at the bottom level) $b_1=v_2$.
        \item Repeat, labeling $t_2=v_3, b_2=v_4, t_3=v_5, b_3=v_6, \ldots,
            t_n=v_{2n-1}, b_n=v_2n$.
        \item For $i=\{1,2,\ldots n\}$, draw $v_i$ at $(i,h(v_i))$.
        \item For $i=\{n+1,n+2, \ldots 2n\}$, draw $v_i$ at $(i-n,h(v_i))$.
        \item Connect $t_1$ and $b_n$ with a vertical straight line.
        \item For all $i\in \{2,3\ldots 2n\}$, connect $v_{i-1}$ to $v_i$ with a
            straight line.
    \end{enumerate}
\end{construction}

We next state two technical lemmas regarding the bowtie construction.

\begin{lemma}[Bowtie Construction is Optimal]\label{lem:bowtiebest}
    Let $G$ be a graph that is a single cycle, with vertices alternating between
    two heights.  Then, the bowtie construction provides a drawing that
    minimizes the number of crossings.
\end{lemma}

\begin{proof}

    Suppose $G$ has $n$ vertices. Because $G$ is a simple cycle, it also has $n$
    edges. By construction, each edge connects a vertex in the first layer with
    a vertex in the second layer.
    Using the bowtie construction, we draw $G$ with $(n-4)/2$ crossings.  In
    particular, two edges have no crossings and $n-2$ edges have one crossing.
    Note that this is the least number of crossings that we could have if there
    are only one or two edges with no crossings.
    So, to draw~$G$ with fewer crossings, we must have at least
    three edges with no crossings.

    \begin{figure}[htb]
        \centering
        \includegraphics{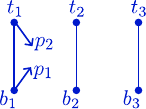}
        \caption{The three uncrossed edges in the proof of \lemref{bowtiebest}.
        The three vertices~$b_1$, $b_2$, and~$b_3$ are at the first height,
        and the three vertices $t_1$, $t_2$, and $t_3$ are at the second height.
        Note that because these three edges are between the same two heights and
        the edges have no crossings, we can think of these three edges with
        $(b_1,t_1)$ the leftmost and $(b_3,t_3)$ the rightmost.  All vertices
        have degree two (because this graph is a cycle).
        }\label{fig:bowtieproof}
    \end{figure}
    Suppose, for the purpose of contradiction, $G$ has a drawing $D$ such that
    there are three edges with no crossing.
    From leftmost in $D$ to rightmost in $D$, let the edges be $(b_1,t_1)$,
    $(b_2,t_2)$, $(b_3,t_3)$, where the first vertex is in the first layer and
    the second vertex in the second layer for each pair, as sketched in
    \figref{bowtieproof}.
    Because~$G$ is a cycle, each vertex has degree
    two. In particular, we think of $t_1$ as having
    two paths leaving it.  Let $p_1$ be the path that visits~$b_1$ next, and let
    $p_2$ be the other path. Note that both paths eventually return to $t_1$,
    after visiting every other vertex exactly once.  We consider the four
    options for where $p_1$ visits next, of the six labeled vertices:
    \begin{itemize}
        \item If $p_2$ visits $b_3$ or $t_3$ before visiting $b_2$ or $t_2$,
            then $p_2$ must cross the edge $(b_2,t_2)$, which is a
            contradiction.  Similarly, $p_1$ must visit $b_2$ or $t_2$ before
            $b_3$ or $t_3$.
        \item Suppose $p_2$ visits $t_2$ before any of the other $b_i$ or $t_i$.
            However, we already established that~$p_1$ must visit $b_1$ or $t_1$
            before $b_3$ or $t_3$.  Because $\deg(b_1) = 2$, we
            know that~$p_2$ must end at $t_1$,  which means that
            the path starting at $b_1$, going up to $t_1$, following~$p_2$ to
            $t_2$, going down the edge to $b_2$, and following $p_1$ backwards
            to $b_1$
            is a cycle
            in $G$ that
            does not contain vertices $t_3$ and $b_3$, which is a contradiction.
        \item If $p_2$ visits $b_2$ before any of the other $b_i$ or $t_i$, then
           we follow a similar logic to find a cycle that does not include $b_3$
            nor $t_3$: starting at $b_1$, going
            up to $t_1$, following $p_2$ to $b_2$, going up the edge to $t_2$,
            and following $p_1$ back to $b_1$, which again is a contradiction.
    \end{itemize}
    Thus, following all cases for where $p_2$ goes leads to a contradiction.
\end{proof}

Now, we sketch the global approach for the general case where $G$ is a cycle.
First, we identify critical vertices at the top and bottom levels
(following Definition~\ref{def:topdown}). These critical vertices globally form
a bowtie-like graph, which we draw using the bowtie technique. But then, the
connections between critical vertices are not single edges but paths, which
means we still need to draw them.

There are two types of subproblems; see \figref{alg-cycle-subproblems}.
The subproblems of Type 1 are single paths that do not interfere with any other
paths, so we can use~\constructref{path} to draw them without crossings. But the
subproblems of Type 2 involve two crossing paths. We now show how to draw any
such pair of crossing paths with only a single crossing.

\begin{figure}[tbh]
    \centering
    \includegraphics[height=1.6in]{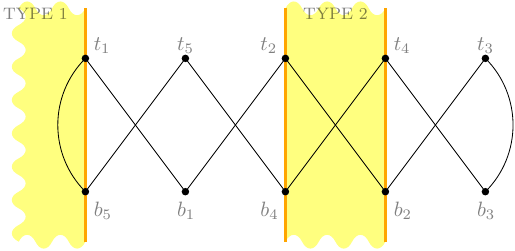}
    \caption{When drawing a global bowtie, there are two types of subproblems to
        consider, as described in the proof of Theorem~\ref{thm:cycle-drawing}.}
    \label{fig:alg-cycle-subproblems}
\end{figure}

\begin{lemma} \label{lem:type2}
      Every Type 2 subproblem has a solution with exactly one crossing.
\end{lemma}
\begin{proof}

    Consider a Type 2 subproblem; that is, four vertices $t_1$, $b_1$, $t_2$,
    and $b_2$ with $h(t_1) = h(t_2) = t$ and $h(b_1) = h(b_2) = b$ and two
    paths, $r$ from $t_1$ to $b_2$ and $g$ from $b_1$ to $t_2$ such that the
    heights of all vertices $v$ of $r$ or $g$ satisfy $b < h(v) < t$.
    To prove the lemma, we need to draw $r$ and $g$ with just a single crossing.

    Identify $r_-$, $r_+$, $g_-$, and $g_+$ as follows:
    \begin{itemize}
      \item $r_-$ is the last vertex of $r$ with height $b$;
      \item $r_+$ is the next point on $r$ next to $r_-$;
      \item $g_-$ is the first vertex of $g$ with height $g$; and
      \item $g_+$ is the previous point on $g$ just before $g_-$
    \end{itemize}
    See \figref {alg-cycle-type2} (a) for an illustration of these vertices.

    We construct a drawing in which the edges $r_-r_+$ and $g_-g_+$ cross, but nothing else does.
    To achieve this, we
    identify four disjoint regions in which to draw the four remaining paths:
    $R^-$ for the part of $r$ up to $r^-$, $R^+$ for the part of $r$ after
    $r^+$, and similar for $G^-$ and~$G^+$. See
    \figref {alg-cycle-type2} (b).

    Now observe that:

    \begin{itemize}
      \item The part of $r$ from $b_1$ to $r_-$ can be drawn without crossings in region $R^-$
      \item The part of $g$ from $t_1$ to $g_+$  can be drawn without crossings in region $G^+$
      \item The part of $r$ from $r^+$ to $t_2$  can be drawn without crossings in region $R^+$
      \item The part of $g$ from $g_-$ to $b_2$ can be drawn without crossings
          in region $G^-$;
    \end{itemize}
    see \figref {alg-cycle-type2} (b, c).
    \journal{\maarten {we might want to put some more details to argue this last
    "observation" in a full version, but I think for now the figure is
    acceptable.}}
\end{proof}

We now have all the ingredients to prove Theorem \ref{thm:cycle-drawing}.

\begin{proof} [Proof of Theorem \ref{thm:cycle-drawing}]

    Given a cycle $G$ and associated map $h$, we first identify its
    top-down-iteration number $k$ and the corresponding key vertices~$t_1,
    \ldots, t_k$ and $b_1, \ldots, b_k$.
    Then, we first layout the subcycle of $G$ consisting only of three key
    vertices using the bow-tie construction. Then we add in the paths between
    the key vertices according to the two types of subproblems:

    To solve subproblems of Type 1, we use \constructref{path}.

    To solve subproblems of Type 2, we use \lemref{type2}.

    \begin{figure}[tbh]
        \centering
        \includegraphics[height=1.3in]{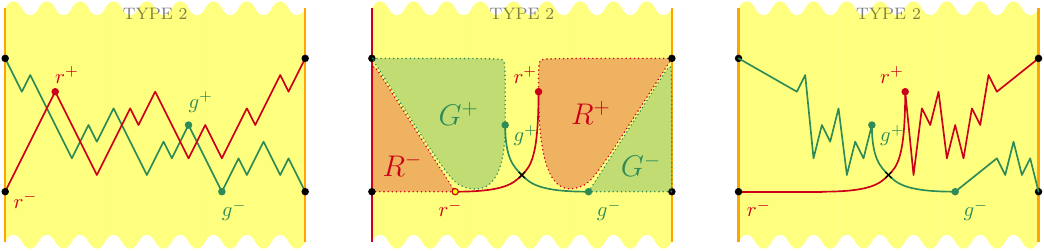}
        \caption{A zoom in of the Type 2 subproblem from the proof of Lemma~\ref{lem:type2}.}
        \label{fig:alg-cycle-type2}
    \end{figure}
\end{proof}

\begin{remark}[One Cycle with Tassels]
    Going beyond drawing single cycle Reeb graphs presents many challenges.
    Even when $G$ has one cycle, but with {\em tassels}---bundles of edges that
    connect one vertex of the cycle to vertices of degree one---we were
    hopeful that a dynamic programming approach would work. The best hope would be
    if we could choose some leftmost edge (perhaps arbitrarily, perhaps try all
    $|E|$, then the remaining vertices can be split into an ``upper portion'' and a
    ``lower portion'' so that we have to calculate: (1) where the
    transition from the ``upper'' and ``lower'' portions is, of which there are $n-1$
    choices; (2) which $x$-order to place the vertices of the upper and lower
    portions, with the hope that a greedy approach works.
    Unfortunately, we have an example that would not be splittable into two such
    portions; see \figref{ex-onecycle-backtracking}, where our tassels would be in place of the green edges.
\end{remark}
\begin{figure}[htb]
    \centering
    \includegraphics{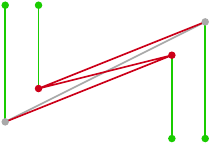}
    \caption{An example that requires backtracking.}\label{fig:ex-onecycle-backtracking}
\end{figure}

\journal{
\section {On Extending the Results}\label{sec:conjectures}
\input{conjectures}
}

\section{Conclusion and Discussion}\label{sec:conclusion}
We investigated the novel problem of drawing Reeb graphs. We established the
computational complexity of minimizing crossing numbers, proving the problem is
NP-hard. Notably, we characterized classes of acyclic Reeb graphs, including
paths and caterpillars, that can be drawn without crossings. For cyclic Reeb
graphs, we derived conditions for crossing-free layouts, particularly for cycles
with unique extrema. We also demonstrated the optimality of the \emph{bowtie}
configuration for cycles with alternating level vertices. Finally, we concluded
by showing that minimizing crossing numbers for general cyclic Reeb graphs is
computationally~feasible.

Several avenues for future research remain. One promising direction involves
developing algorithms for drawing Reeb graphs with arbitrary tree structures,
extending beyond the acyclic cases explored here. The complexity of drawing Reeb
graphs of surfaces with higher genus seems significantly more complex; even a
cycle with extra edges is nontrivial, as described in
Section~\ref{sec:cycles}.

\bibliographystyle{plainurl}
\bibliography{refs}


\end{document}